\documentclass[aps,onecolumn,showpacs,groupedaddress]{revtex4-1}
\usepackage{amssymb}
\usepackage{mathrsfs}
\usepackage{pstricks}
\usepackage{color}
\usepackage{graphicx}
\usepackage[fleqn]{amsmath}
\usepackage{slashed}
\usepackage{amsmath}
\usepackage{mathtools}
\usepackage{array}
\usepackage{tabularx}
\usepackage{physics}
\usepackage{slashed}
\usepackage[utf8]{inputenc}
\usepackage[T1]{fontenc}
\usepackage{float}
\usepackage{amsthm}
\newtheorem*{theorem}{Theorem}
\usepackage{physics}

\begin{document}


\title{\bf Exact infrared scaling behavior of Randers-Finsler scalar field theories}


\author{M. S. Mendes}
\author{J. F. S. Neto}
\author{R. F. Silva}
\author{H. A. S. Costa}
\author{P. R. S. Carvalho}
\email{prscarvalho@ufpi.edu.br}
\affiliation{\it Departamento de F\'\i sica, Universidade Federal do Piau\'\i, 64049-550, Teresina, PI, Brazil}





\begin{abstract}
We study the scaling behavior of Randers-Finsler massless scalar field theories in the infrared regime. For that, we compute analytically the radiative corrections to the corresponding anomalous dimensions, related to the critical exponents of the theory, first up to next-to-leading loop order and later for all-loop levels. We consider the effect of the Randers-Finsler space-time properties on the critical exponents by considering the parameter characterizing those space-times in its exact form. We employ field-theoretic renormalization group and $\epsilon$-expansion techniques at dimensions $d = 4 - \epsilon$ through three distinct and independent methods. At the end, we furnish the physical interpretation of the obtained results. 
\end{abstract}


\maketitle


\section{Introduction} 

\par Quantum fields in Randers-Finsler space-times \cite{Bejancu} have been studied in recent years \cite{10.1063.5.0065944,Silva2021}. The quantum aspects of fields embedded in those space-times were approached in the case of free fields \cite{Silva2021}. The Finsler geometry is characterized by some length interval, namely, $ds_{R} = F_{R}(x,\dot{x})dt$, where the general function $F_{R}(x,\dot{x})$ is called Finsler function \cite{Bejancu}. Under the appropriate limit, the Finsler geometry specializes to the Riemann one. For some given Finsler function, we have the Randers-Finsler space-time \cite{Silva2021}, \emph{i. e.}, $ds_{R} = (\sqrt{-g_{\mu\nu}\dot{x}^{\mu}\dot{x}^{\nu}} + \zeta a_{\mu}(x)\dot{x}^{\mu})dt$. This specific length interval is composed of the standard Lorentz invariant interval and a background vector $a_{\mu}$. The $\zeta$ parameter indicates the deviation of Lorentz symmetry preservation and breaks that symmetry. In practical calculations we can consider small deviations of the Lorentz symmetry by taking small values of $\zeta$ in which $\zeta^{2}g^{\mu\nu}a_{\mu}a_{\nu} \ll 1$ for space-like $a_{\mu}$ and $\zeta^{2}g^{\mu\nu}a_{\mu}a_{\nu} \gg 1$ for time-like $a_{\mu}$ \cite{Silva2021}. In this work we have to consider $\zeta$ in its exact form. Applications of Randers-Finsler space-times have been shown in astrophysics and cosmology \cite{VACARU2010224,PRAVEEN2025100030,KapsabelisE,PhysRevD.88.123510,PhysRevD.87.043506,CHANG2009173}.  

\par In this work we address the problem of considering the influence of the properties of such space-times on the anomalous dimensions of massless self-interacting O($N$) $\lambda\phi^{4}$ scalar field theory. These anomalous dimensions are, in turn, related to the critical exponents of the theory which dictates the scaling properties of the corresponding correlation functions. We compute effects beyond the free field and three-level regime by considering radiative corrections up to next-to-leading order (NLO). For that we employ field-theoretic perturbative renormalization group and $\epsilon$-expansion techniques at dimensions $d = 4 - \epsilon$ through three distinct and independent methods. In considering the parameter characterizing Randers-Finsler space-times, we first consider the evaluation of Feynman integrals of the theory in powers of this parameter. As this procedure leads to a tedious aim, we propose later to take into account this parameter in its exact form for NLO. We then display similar calculations by treating the parameter in its exact form and valid for any loop level.  

\par As it is known \cite{ZinnJustin,Amit}, the critical exponents are universal quantities. In the context of both fluids and magnetic systems, the critical behavior of these two kinds of systems have the same set of critical exponent. This fact can be understood once in the critical behavior of systems undergoing a continuous phase transition, the critical exponents do not depend on their microscopic details as the form of their lattices or their specific critical temperature values. They depend only on their dimension $d$, number $N$ of some $N$-component order parameter (spin dimension for magnetic systems) and the interactions among the constituents of the systems. We have to study systems belonging to some general O($N$) universality class whose systems are characterized by same specific values of the number of fields $N$, \emph{i. e.}: Ising ($N = 1$), XY ($N = 2$), Heisenberg ($N = 3$), self-avoiding random walk ($N = 0$), spherical ($N \rightarrow \infty$) etc \cite{Pelissetto2002549}. Once the critical exponents are universal quantities, their results as ones obtained through distinct and independent methods must be the same, although some quantities as the corresponding renormalization constants, $\beta$-functions, anomalous dimensions and fixed points present different expressions. The critical exponents furnish the the scaling behavior of the primitively one-particle irreducible vertex parts ($1$PI) of the theory, namely the $\Gamma^{(2)}$, $\Gamma^{(4)}$ and $\Gamma^{(2,1)}$. These vertex parts are related to the correlation functions of the theory and express all the divergent properties necessary to determine the critical exponents. As we shall see, by removing the divergences of the theory, we can determine finite $\beta$-functions and anomalous dimensions which can be used for determining the critical exponents. This aims is attained by computing the anomalous dimensions at the fixed point which is obtained as the nontrivial root of the $\beta$-function.  

\par In this work we compute the effect of the Finsler-Randers space-time properties on the critical exponents of the theory defined in Sec. \ref{Randers-Finsler scalar field theory} through Sec. \ref{Normalization conditions method} through the Normalization conditions method. The same task is employed in Sec. \ref{Minimal subtraction scheme} but now in Minimal subtraction scheme. In Sec. \ref{Massless BPHZ method} we attain the same aim by applying the Massless BPHZ method. In Sec. \ref{Critical exponents to any loop order} we compute the critical exponents values valid for any loop levels. In Sec. \ref{Conclusions} we present our conclusions.

\section{Randers-Finsler scalar field theory}\label{Randers-Finsler scalar field theory}

\par The massive scalar field defined at Randers-Finsler space-time is governed by the following dispersion relation \cite{Silva2021}
\begin{eqnarray}
P^{2} + \zeta^{2}a_{\mu}a_{\nu}P^{\mu}P^{\nu} - 2\zeta ma_{\mu}P^{\mu} = - m^{2}.
\end{eqnarray}
So for massless scalar fields, we can write the corresponding Lagrangian density as \cite{Silva2021} 
\begin{eqnarray}\label{huytrji}
\mathcal{L}_{B,\zeta_{a}} = \frac{1}{2}\partial^{\mu}\phi_{B}\partial_{\mu}\phi_{B} + \frac{1}{2}
[(\zeta_{a}\cdot\partial)\phi_{B}]^{2} + \frac{\lambda_{B}}{4!}\phi_{B}^{4} +  \frac{1}{2}t_{B}\phi_{B}^{2},
\end{eqnarray}
in the Euclidean metric signature ($++++$) suitable for critical exponents computation \cite{Amit}, where $\zeta_{a}\equiv \zeta a$. The parameter $\zeta$ and the four-vector $a^{\mu}$ are responsible for taking into account to the properties of the Randers-Finsler space-time. In the limit $\zeta\rightarrow 0$, we recover the Euclidean space-time structure \cite{Amit}. We consider the bare Lagrangian density from the beginning as a bare one once the corresponding quantities are divergent, namely the bare $N$-component field $\phi_{B}$, coupling constant $\lambda_{B}$ and composite field coupling constant $t_{B}$. As any physical theory plagued by divergences does not make sense, we have to get rid such divergences. For that, we employ some renormalization methods. Now we proceed to apply them. 

\section{Normalization conditions method}\label{Normalization conditions method}

\par The normalization condition method \cite{Amit} is characterized by one getting rid the divergences of the theory, where we start from the bare theory. We apply multiplicative renormalization in the primitively one-particle irreducible ($1$PI) vertex parts 
\begin{eqnarray}
\Gamma_{\zeta_{a}}^{(n, l)}(P_{i}, Q_{j}, g, \kappa) = Z_{\phi,\zeta_{a}}^{n/2}Z_{\phi^{2},\zeta_{a}}^{l}\Gamma_{B,\zeta_{a}}^{(n, l)}(P_{i}, Q_{j}, \lambda_{B}),
\end{eqnarray}
where ($i = 1, \cdots, n$, $j = 1, \cdots, l$). We keep their external momenta at fixed values through the set of the following normalization conditions
\begin{eqnarray}\label{ygfdxzsze}
\Gamma_{\zeta_{a}}^{(2)}[P^{2} + (\zeta_{a}\cdot P)^{2} = 0, g] = 0, 
\end{eqnarray}
\begin{eqnarray}
\frac{\partial \Gamma_{\zeta_{a}}^{(2)}[P^{2} + (\zeta_{a}\cdot P)^{2}, g]}{\partial [P^{2} + (\zeta_{a}\cdot P)^{2}]}\Biggr|_{P^{2} + (\zeta_{a}\cdot P)^{2} = \kappa^{2}} = 1,
\end{eqnarray}
\begin{eqnarray}\label{jijhygtfrd}
\Gamma_{\zeta_{a}}^{(4)}[P^{2} + (\zeta_{a}\cdot P)^{2}, g]|_{SP} = g, 
\end{eqnarray}
\begin{eqnarray}
\Gamma_{\zeta_{a}}^{(2,1)}(P_{1}, P_{2}, Q_{3}, g)|_{\overline{SP}} = 1,
\end{eqnarray}
where $Q_{3} = -(P_{1} + P_{2})$ for the symmetry point (SP): $P_{i}\cdot P_{j} = (\kappa^{2}/4)(4\delta_{ij}-1)$, which implies that $(P_{i} + P_{j})^{2} \equiv P^{2} + (\zeta_{a}\cdot P)^{2} = \kappa^{2}$ for $i\neq j$. For another symmetry point $\overline{SP}$: $P_{i}^{2} = 3\kappa^{2}/4$ and $P_{1}\cdot P_{2} = -\kappa^{2}/4$, by implying $(P_{1} + P_{2})^{2} \equiv P^{2} + (\zeta_{a}\cdot P)^{2} = \kappa^{2}$. In terms of $\kappa$ units, we have that $P^{2} + (\zeta_{a}\cdot P)^{2} = \kappa^{2} \rightarrow 1$. The parameter $\kappa$ is some arbitrary momentum scale parameter. Also we define the dimensionless bare $u_{0}$ and renormalized $u$ coupling constants as $\lambda_{B} = u_{B}\kappa^{\epsilon/2}$ and $g = u\kappa^{\epsilon/2}$, respectively, in dimensions $d = 4 - \epsilon$. The primitively $1$PI vertex parts, up to next-to-leading order, are given by
\begin{eqnarray}\label{gtfrdrdes}
\Gamma^{(2)}_{B,\zeta_{a}} = \quad \parbox{12mm}{\includegraphics[scale=1.0]{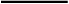}}^{-1} + \frac{1}{6}\hspace{1mm}\parbox{12mm}{\includegraphics[scale=1.0]{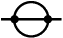}} + \frac{1}{4}\hspace{1mm}\parbox{10mm}{\includegraphics[scale=0.8]{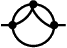}} ,
\end{eqnarray}
\begin{eqnarray}
\Gamma^{(4)}_{B,\zeta_{a}} = \quad \parbox{12mm}{\includegraphics[scale=0.09]{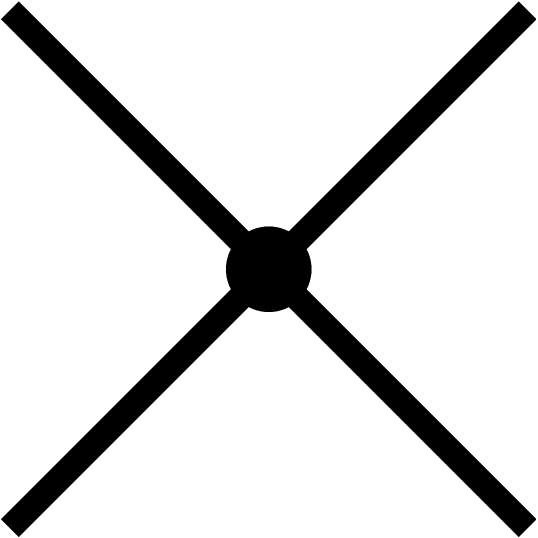}} + \quad \frac{1}{2}\hspace{1mm}\parbox{10mm}{\includegraphics[scale=1.0]{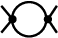}} + 2 \hspace{1mm} perm. +  \frac{1}{4}\hspace{1mm}\parbox{16mm}{\includegraphics[scale=1.0]{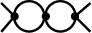}} + 2 \hspace{1mm} perm. +  \frac{1}{2}\hspace{1mm}\parbox{12mm}{\includegraphics[scale=0.8]{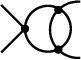}} + 5 \hspace{1mm} perm. , 
\end{eqnarray}
\begin{eqnarray}\label{gtfrdesuuji}
\Gamma^{(2,1)}_{B,\zeta_{a}} = \quad \parbox{14mm}{\includegraphics[scale=1.0]{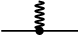}} \quad + \quad \frac{1}{2}\hspace{1mm}\parbox{14mm}{\includegraphics[scale=1.0]{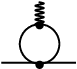}} \quad +  \frac{1}{4}\hspace{1mm}\parbox{12mm}{\includegraphics[scale=1.0]{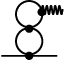}}  + \quad \frac{1}{2}\hspace{1mm}\parbox{12mm}{\includegraphics[scale=0.8]{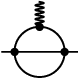}} ,
\end{eqnarray}
where the term \textit{perm.} means a permutation of the external momenta attached to the cut external lines. All internal lines are represented by the following free propagator 
\begin{eqnarray}
G_{0,\zeta_{a}}^{-1}(q) = \parbox{12mm}{\includegraphics[scale=1.0]{fig9.eps}}^{-1} = q^{2} + (\zeta_{a}\cdot q)^{2}.
\end{eqnarray}
We can write the renormalization constants $Z_{\phi,\zeta_{a}}$ and $\overline{Z}_{\phi^{2},\zeta_{a}} \equiv Z_{\phi,\zeta_{a}}Z_{\phi,\zeta_{a}^{2}}$  perturbatively as powers of $u$ and obtain the $\beta_{\zeta_{a}}$-function and anomalous dimensions $\gamma_{\phi,\zeta_{a}}$ and the composite field anomalous dimension $\gamma_{\phi,\zeta_{a}^{2}}$ through the renornalization group equation given by
\begin{eqnarray}
\left( \kappa\frac{\partial}{\partial\kappa} + \beta_{\zeta_{a}}\frac{\partial}{\partial u} - \frac{1}{2}n\gamma_{\phi,\zeta_{a}} + l\gamma_{\phi,\zeta_{a}^{2}} \right)\Gamma_{\zeta_{a}}^{(n, l)} = 0,
\end{eqnarray}
where 
\begin{eqnarray}\label{kjsjffbxdzs}
\beta_{\zeta_{a}}(u) = \kappa\frac{\partial u}{\partial \kappa} = -\epsilon\left(\frac{\partial\ln u_{0}}{\partial u}\right)^{-1},
\end{eqnarray}
\begin{eqnarray}
\gamma_{\phi,\zeta_{a}}(u) = \beta(u)_{\zeta_{a}}\frac{\partial\ln Z_{\phi,\zeta_{a}}}{\partial u},
\end{eqnarray}
\begin{eqnarray}
\gamma_{\phi^{2},\zeta_{a}}(u) = -\beta_{\zeta_{a}}(u)\frac{\partial\ln Z_{\phi^{2},\zeta_{a}}}{\partial u}.
\end{eqnarray}
Instead employing the composite field anomalous dimension $\gamma_{\phi^{2},\zeta_{a}}$, we apply the following one 
\begin{eqnarray}
\overline{\gamma}_{\phi^{2},\zeta_{a}}(u) = -\beta_{\zeta_{a}}(u)\frac{\partial\ln \overline{Z}_{\phi^{2},\zeta_{a}}}{\partial u} \equiv \gamma_{\phi^{2},\zeta_{a}}(u) - \gamma_{\phi,\zeta_{a}}(u)
\end{eqnarray}
for convenience reasons. 

\par The Feynman integrals computation, for considering the effect of the Randers-Finsler space-time, would involve to consider the expansion in the small parameter $\zeta$, up to fourth order for example, in the form 
\begin{eqnarray}
\dfrac{1}{q^{2} + (\zeta_{a}\cdot q)^{2}} =  \dfrac{1}{q^{2}}\left[1 - \dfrac{(\zeta_{a}\cdot q)^{2}}{q^{2}} + \dfrac{(\zeta_{a}\cdot q)^{4}}{(q^{2})^{2}} \right].
\end{eqnarray}
Now by applying dimensional regularization 
\begin{eqnarray}
\int \frac{d^{d}q}{(2\pi)^{d}} \frac{1}{(q^{2} + 2pq + M^{2})^{\alpha}} =     \hat{S}_{d}\frac{1}{2}\frac{\Gamma(d/2)}{\Gamma(\alpha)}\frac{\Gamma(\alpha - d/2)}{(M^{2} - p^{2})^{\alpha - d/2}},
\end{eqnarray}

\begin{eqnarray}
\int \frac{d^{d}q}{(2\pi)^{d}} \frac{q^{\mu}}{(q^{2} + 2pq + M^{2})^{\alpha}} =     -\hat{S}_{d}\frac{1}{2}\frac{\Gamma(d/2)}{\Gamma(\alpha)}\frac{p^{\mu}\Gamma(\alpha - d/2)}{(M^{2} - p^{2})^{\alpha - d/2}},
\end{eqnarray}

\begin{eqnarray}
\int \frac{d^{d}q}{(2\pi)^{d}} \frac{q^{\mu}q^{\nu}}{(q^{2} + 2pq + M^{2})^{\alpha}} =   \hat{S}_{d}\frac{1}{2}\frac{\Gamma(d/2)}{\Gamma(\alpha)} \left[ \frac{1}{2}\delta^{\mu\nu}\frac{\Gamma(\alpha - 1 - d/2)}{(M^{2} - p^{2})^{\alpha - 1 - d/2}}  +  p^{\mu}p^{\nu}\frac{\Gamma(\alpha - d/2)}{(M^{2} - p^{2})^{\alpha - d/2}} \right],\nonumber \\&& 
\end{eqnarray}

\begin{eqnarray}
&&\int \frac{d^{d}q}{(2\pi)^{d}} \frac{q^{\mu}q^{\nu}q^{\rho}}{(q^{2} + 2pq + M^{2})^{\alpha}} = - \hat{S}_{d}\frac{1}{2}\frac{\Gamma(d/2)}{\Gamma(\alpha)}\times \nonumber \\&&    \left[\frac{1}{2}[\delta^{\mu\nu}p^{\rho} + \delta^{\mu\rho}p^{\nu} + \delta^{\nu\rho}p^{\mu}]\frac{\Gamma(\alpha - 1 - d/2)}{(M^{2} - p^{2})^{\alpha - 1 - d/2}}  + \ p^{\mu}p^{\nu}p^{\rho}\frac{\Gamma(\alpha - d/2)}{(M^{2} - p^{2})^{\alpha - d/2}} \right],\nonumber \\
\end{eqnarray}

\begin{eqnarray}
&&\int \frac{d^{d}q}{(2\pi)^{d}} \frac{q^{\mu}q^{\nu}q^{\rho}q^{\sigma}}{(q^{2} + 2pq + M^{2})^{\alpha}} = \hat{S}_{d}\frac{1}{2}\frac{\Gamma(d/2)}{\Gamma(\alpha)}\times  \nonumber \\&&  \left\{\frac{1}{4}[\delta^{\mu\nu}\delta^{\rho\sigma} + \delta^{\mu\rho}\delta^{\nu\sigma} +\delta^{\mu\sigma}\delta^{\nu\rho}]\frac{\Gamma(\alpha - 2 - d/2)}{(M^{2} - p^{2})^{\alpha - 2 - d/2}}  \right.  \nonumber \\  &&\left. + \frac{1}{2}[\delta^{\mu\nu}p^{\rho}p^{\sigma} + \delta^{\mu\rho}p^{\nu}p^{\sigma} + \delta^{\mu\sigma}p^{\nu}p^{\rho}   + \delta^{\nu\rho}p^{\mu}p^{\sigma} +\delta^{\nu\sigma}p^{\mu}p^{\rho} +\delta^{\rho\sigma}p^{\mu}p^{\nu}]\frac{\Gamma(\alpha - 1 - d/2)}{(M^{2} - p^{2})^{\alpha - 1 - d/2}}   \right.  \nonumber \\  &&\left. + p^{\mu}p^{\nu}p^{\rho}p^{\sigma}\frac{\Gamma(\alpha - d/2)}{(M^{2} - p^{2})^{\alpha - d/2}} \right\},
\end{eqnarray}
where $\hat{S}_{d} \equiv S_{d}/(2\pi)^{d} = [2^{d-1}\pi^{d/2}\Gamma(d/2)]^{-1}$ and $S_{d} = 2\pi^{d/2}/\Gamma(d/2)$ is the unit $d$-dimensional sphere area, we obtain for the diagram
\begin{eqnarray}\label{okjhfy}
\parbox{12mm}{\includegraphics[scale=1.0]{fig10.eps}} = \int \dfrac{d^{d}q}{(2\pi)^{d}}\dfrac{1}{q^{2} + (\zeta_{a}\cdot q)^{2}} \dfrac{1}{(q + P)^{2} + [\zeta_{a}\cdot (q + P)]^{2}},
\end{eqnarray}  
the following result, in $d = 4 - \epsilon$ dimensions, 
\begin{eqnarray}
\parbox{10mm}{\includegraphics[scale=1.0]{fig10.eps}}_{SP} = \frac{Z_{\zeta_{a}}}{\epsilon}\left(1 + \frac{1}{2}\epsilon  \right),
\end{eqnarray}  
where
\begin{eqnarray}\label{pfhg}
Z_{\zeta_{a}} = 1 - \frac{1}{2}\zeta_{a}^{2} + \frac{3}{8}\zeta_{a}^{4}.
\end{eqnarray}
We could proceed in computing the remaining Feynman diagrams up to higher orders in $\zeta$ but this would lead to a tedious computation process. As we must show now is that it is possible to compute the Feynman diagrams by considering the effect of the Randers-Finsler space-time in exact form, \emph{i. e.}, by treating the parameter $\zeta$ in an exact way.

\par Consider, for example, the Feynman diagram of Eq. (\ref{okjhfy}). For some arbitrary moment, the inverse free propagator can be written as
\begin{eqnarray}
q^{2} + (\zeta_{a}\cdot q)^{2} = (\delta_{\mu\nu} + \zeta^{2}a_{\mu}a_{\nu})q^{\mu}q^{\nu} = q^{t}(\mathbb{I} + \zeta_{a}\zeta_{a}^{t})q,~
\end{eqnarray}
where the momentum $q$ ($\zeta_{a}$) is a $d$-dimensional vector expressed by some column matrix and $q^{t}$ ($\zeta_{a}^{t}$) is the corresponding row matrix. The $\mathbb{I}$ and $\zeta_{a}\zeta_{a}^{t}$ matrices are the referred matrix representations of both identity and $(\zeta_{a}\zeta_{a}^{t})_{\mu\nu} = \zeta^{2}a_{\mu}a_{\nu}$ matrices. Now by applying the transformation
\begin{eqnarray}
q^{\prime} \rightarrow \sqrt{\mathbb{I} + \zeta_{a}\zeta_{a}^{t}}~q,  
\end{eqnarray}
the $d$-dimensional volume element of the Feynman integral transforms as
\begin{eqnarray}
d^{d}q^{\prime} = \sqrt{\det(\mathbb{I} + \zeta_{a}\zeta_{a}^{t})}\,d^{d}q.
\end{eqnarray}
Thus the integral of of Eq. (\ref{okjhfy}), after making the transformation $q^{\prime}\rightarrow q$, turns out to be
\begin{eqnarray}\label{yokjhfy}
\parbox{12mm}{\includegraphics[scale=1.0]{fig10.eps}} = \mathbf{Z}_{\zeta_{a}}\int \dfrac{d^{d}q}{(2\pi)^{d}}\dfrac{1}{q^{2}(q + P)^{2}},
\end{eqnarray}
where
\begin{eqnarray}\label{iutyf}
\mathbf{Z}_{\zeta_{a}} = \dfrac{1}{\sqrt{\det(\mathbb{I} + \zeta_{a}\zeta_{a}^{t})}}.
\end{eqnarray}
By expanding the Eq. (\ref{iutyf}) in powers of $\zeta$ when $\zeta\ll 1$, we obtain the same result as that of Eq. (\ref{pfhg})  up to fourth order in $\zeta$, as expected, namely
\begin{eqnarray}
\mathbf{Z}_{\zeta_{a}} = 1 - \frac{1}{2}\zeta_{a}^{2} + \frac{3}{8}\zeta_{a}^{4} + ....
\end{eqnarray}
Now we can compute the other Feynman diagrams needed for the determination of the critical exponents up to NLO \cite{Amit} by considering the parameter $\zeta$ exactly and obtain 
\begin{eqnarray}
\parbox{10mm}{\includegraphics[scale=1.0]{fig10.eps}}_{SP} = \frac{\mathbf{Z}_{\zeta_{a}}}{\epsilon}\left(1 + \frac{1}{2}\epsilon  \right), 
\end{eqnarray}   
\begin{eqnarray}
\parbox{12mm}{\includegraphics[scale=1.0]{fig6.eps}}^{\prime} = -\frac{\mathbf{Z}_{\zeta_{a}}}{8\epsilon}\left( 1 + \frac{5}{4}\epsilon \right)\mathbf{\Pi}^{2},
\end{eqnarray}  
\begin{eqnarray}
\parbox{10mm}{\includegraphics[scale=0.9]{fig7.eps}}^{\prime} = -\frac{\mathbf{Z}_{\zeta_{a}}}{6\epsilon^{2}}\left( 1 + 2\epsilon  \right)\mathbf{\Pi}^{3},
\end{eqnarray}  
\begin{eqnarray}
\parbox{12mm}{\includegraphics[scale=0.8]{fig21.eps}}_{SP} = \frac{\mathbf{Z}_{\zeta_{a}}}{2\epsilon^{2}}\left( 1 + \frac{3}{2}\epsilon \right)\mathbf{\Pi}^{2},
\end{eqnarray}  
where (in units of $\kappa$), 
\begin{eqnarray}
\parbox{10mm}{\includegraphics[scale=1.0]{fig10.eps}}_{SP} \equiv \parbox{10mm}{\includegraphics[scale=1.0]{fig10.eps}}\bigg\vert_{P^{2} + (\zeta_{a}\cdot P)^{2} = 1}
\end{eqnarray} 
\begin{eqnarray}
\parbox{10mm}{\includegraphics[scale=1.0]{fig6.eps}}^{\prime} \equiv \dfrac{\partial \parbox{10mm}{\includegraphics[scale=1.0]{fig6.eps}}}{\partial [P^{2} + (\zeta_{a}\cdot P)^{2}]}\bigg\vert_{P^{2} + (\zeta_{a}\cdot P)^{2} = 1}
\end{eqnarray}  
\begin{eqnarray}
\parbox{12mm}{\includegraphics[scale=0.8]{fig21.eps}}_{SP} \equiv \parbox{10mm}{\includegraphics[scale=1.0]{fig10.eps}}\bigg\vert_{P^{2} + (\zeta_{a}\cdot P)^{2} = 1}
\end{eqnarray} 
\begin{eqnarray}
\parbox{10mm}{\includegraphics[scale=0.9]{fig7.eps}}^{\prime} \equiv \dfrac{\partial \parbox{10mm}{\includegraphics[scale=0.9]{fig7.eps}}}{\partial [P^{2} + (\zeta_{a}\cdot P)^{2}]}\bigg\vert_{P^{2} + (\zeta_{a}\cdot P)^{2} = 1}
\end{eqnarray}  
Now we can obtain the $\beta$-function and the anomalous dimensions for any value of $N$ \cite{Amit}, thus obtaining 
\begin{eqnarray}\label{fhufhudh}
\beta_{\zeta_{a}}(u) = -\epsilon u +  \mathbf{Z}_{\zeta_{a}}\frac{N + 8}{6}\left( 1 + \frac{1}{2}\epsilon \right) u^{2} -  \mathbf{Z}_{\zeta_{a}}^{2}\frac{3N + 14}{12}, 
\end{eqnarray}
\begin{eqnarray}\label{gkjlhitu}
\gamma_{\phi,\zeta_{a}}(u) = \mathbf{Z}_{\zeta_{a}}^{2}\frac{N + 2}{72}\left( 1 + \frac{5}{4}\epsilon \right)u^{2} -  \mathbf{Z}_{\zeta_{a}}^{3}\frac{(N + 2)(N + 8)}{864}u^{3},  
\end{eqnarray}
\begin{eqnarray}\label{dkvyenh}
\overline{\gamma}_{\phi^{2},\zeta_{a}}(u) = \mathbf{Z}_{\zeta_{a}}\frac{N + 2}{6}\left( 1 + \frac{1}{2}\epsilon \right) u -  \mathbf{Z}_{\zeta_{a}}^{2}\frac{N + 2}{12}u^{2}.
\end{eqnarray}
Now we can obtain the critical exponents of the theory. For that we have to compute the nontrivial solution for the condition
\begin{eqnarray}
\beta_{\zeta_{a}}(u_{\zeta_{a}}^{*}) = 0, 
\end{eqnarray}
whose value is given by 
\begin{eqnarray}
u_{\zeta_{a}}^{*} = \frac{6\epsilon/\mathbf{Z}_{\zeta_{a}}}{(N + 8)} \left\{ 1 + \epsilon\left[ \frac{3(3N + 14)}{(N + 8)^{2}} -\frac{1}{2} \right]\right\}.
\end{eqnarray}
When we apply the relations $\eta_{\zeta_{a}}\equiv\gamma_{\phi,\zeta_{a}}(u_{\zeta_{a}}^{*})$ and $\nu_{\zeta_{a}}^{-1}\equiv 2 - \eta_{\zeta_{a}} - \overline{\gamma}_{\phi^{2},\zeta_{a}}(u_{\zeta_{a}}^{*})$
we find the following results 
\begin{eqnarray}\label{eta}
\eta_{\zeta_{a}} = \frac{(N + 2)\epsilon^{2}}{2(N + 8)^{2}}\left\{ 1 + \epsilon\left[ \frac{6(3N + 14)}{(N + 8)^{2}} -\frac{1}{4} \right]\right\} \equiv\eta ,
\end{eqnarray}
\begin{eqnarray}\label{nu}
\nu_{\zeta_{a}} = \frac{1}{2} + \frac{(N + 2)\epsilon}{4(N + 8)} +  \frac{(N + 2)(N^{2} + 23N + 60)\epsilon^{2}}{8(N + 8)^{3}} \equiv\nu ,
\end{eqnarray}
where $\eta$ and $\nu$ are the critical exponents for the theory in Euclidean space-time \cite{Wilson197475,PhysRevLett.28.240,PhysRevLett.28.548}, \emph{i. e.}, one for which $\zeta = 0$. The physical interpretation of this result is as follows. The properties we are considering are that of the space-time where the field is embedded and they do not modify how the field field interacts with itself (which would lead to a modification of the critical exponents values), although the $\beta_{\zeta_{a}}$-function and anomalous dimensions depend on $\zeta_{a}$. This result is in accordance with the universality hypothesis, which asserts that the critical exponents values depend only on the dimension $d$, number $N$ of some $N$-component order parameter and the interactions among the constituents of the system. Now we compute the critical exponents by employing a distinct and independent method. 

\section{Minimal subtraction scheme}\label{Minimal subtraction scheme}

\par The next method to be approached is the minimal subtraction scheme one \cite{Amit}. As in the earlier method we start from the bare theory but, on the other hand, now we let the external momenta values arbitrary and do not fix them at specific values as in the normalization conditions method. This fact displays the generality and elegance of the present method. Then, by computing the same diagrams as the computed in the early Sec. but now with their external momenta kept at arbitrary values we obtain
\begin{eqnarray}
\parbox{10mm}{\includegraphics[scale=1.0]{fig10.eps}} = \frac{\mathbf{Z}_{\zeta_{a}}}{\epsilon} \left\{1 - \frac{1}{2}\epsilon - \frac{1}{2}\epsilon L[P^{2} + (\zeta_{a}\cdot P)^{2}] \right\},
\end{eqnarray}   
\begin{eqnarray}
\parbox{12mm}{\includegraphics[scale=1.0]{fig6.eps}} = -\mathbf{Z}_{\zeta_{a}}^{2}\frac{P^{2} + (\zeta_{a}\cdot P)^{2}}{8\epsilon}\left\{ 1 + \frac{1}{4}\epsilon -2\epsilon L_{3}[P^{2} + (\zeta_{a}\cdot P)^{2}] \right\},
\end{eqnarray}  
\begin{eqnarray}
\parbox{10mm}{\includegraphics[scale=0.9]{fig7.eps}} = -\mathbf{Z}_{\zeta_{a}}^{3}\frac{P^{2} + (\zeta_{a}\cdot P)^{2}}{6\epsilon^{2}}\left\{ 1 + \frac{1}{2}\epsilon -3\epsilon L_{3}[P^{2} + (\zeta_{a}\cdot P)^{2}] \right\},
\end{eqnarray}  
\begin{eqnarray}
\parbox{12mm}{\includegraphics[scale=0.8]{fig21.eps}} = \frac{\mathbf{Z}_{\zeta_{a}}^{2}}{2\epsilon^{2}}\left\{1 - \frac{1}{2}\epsilon - \epsilon L[P^{2} + (\zeta_{a}\cdot P)^{2}] \right\},
\end{eqnarray}
where
\begin{eqnarray}\label{uhduhguh}
L[P^{2} + (\zeta_{a}\cdot P)^{2}] = \int_{0}^{1}dx\ln\{x(1-x)[P^{2} + (\zeta_{a}\cdot P)^{2}]\},
\end{eqnarray}
\begin{eqnarray}\label{uhduhguhf}
L_{3}[P^{2} + (\zeta_{a}\cdot P)^{2}] =  \int_{0}^{1}dx(1-x)\ln\{x(1-x)[P^{2} + (\zeta_{a}\cdot P)^{2}]\},
\end{eqnarray}
The $\beta_{\zeta_{a}}$-function and anomalous dimensions now assume the form
\begin{eqnarray}\label{uahuahuahu}
\beta_{\zeta_{a}}(u) = -\epsilon u + \mathbf{Z}_{\zeta_{a}}\frac{N + 8}{6}u^{2} - \mathbf{Z}_{\zeta_{a}}^{2}\frac{3N + 14}{12}u^{3},
\end{eqnarray}
\begin{eqnarray}
\gamma_{\phi,\zeta_{a}}(u) = \mathbf{Z}_{\zeta_{a}}^{2}\frac{N + 2}{72}u^{2} - \mathbf{Z}_{\zeta_{a}}^{3}\frac{(N + 2)(N + 8)}{1728}u^{3},
\end{eqnarray}
\begin{eqnarray}\label{uahuahuahuaa}
\overline{\gamma}_{\phi^{2},\zeta_{a}}(u) = \mathbf{Z}_{\zeta_{a}}\frac{N + 2}{6} u - \mathbf{Z}_{\zeta_{a}}^{2}\frac{N + 2}{12}u^{2}.
\end{eqnarray}
The external momentum-dependent integrals of Eqs. (\ref{uhduhguh})-(\ref{uhduhguhf}) have canceled out in the renormalization process of computation of the $\beta_{\zeta_{a}}$-function and anomalous dimensions ans must not be evaluated. This shows its elegance and generality, once we do not have to choose any fixed specific values for the external momenta. Now the nontrivial fixed point is given by 
\begin{eqnarray}
 u_{\zeta_{a}}^{*} = \frac{6\epsilon/\mathbf{Z}_{\zeta_{a}}}{(N + 8)}\left\{ 1 + \epsilon\left[ \frac{3(3N + 14)}{(N + 8)^{2}} \right]\right\}.
\end{eqnarray}
When we compute the critical exponents, from $\eta_{\zeta_{a}}\equiv\gamma_{\phi,\zeta_{a}}(u_{\zeta_{a}}^{*})$ and $\nu_{\zeta_{a}}^{-1}\equiv 2 - \eta_{\zeta_{a}} - \overline{\gamma}_{\phi^{2},\zeta_{a}}(u_{\zeta_{a}}^{*})$, we attain the same results as the ones of the earlir Sec., \emph{i. e.}, that the critical exponents values obtained are the same as those from the Euclidean theory \cite{Wilson197475,PhysRevLett.28.240,PhysRevLett.28.548}. This procedure of obtaining the same result employing distinct and independent methods reinforces the arbitrariness of the field theoretic renormalization group scheme we can adopt. Finally we have apply a third renormalization scheme in next Sect..

\section{Massless BPHZ method}\label{Massless BPHZ method}

\par In the present method, the BPHZ (Bogoliubov-Parasyuk-Hepp-Zimmermann) one \cite{BogoliubovParasyuk,Hepp,Zimmermann}, we start from the renormalized Lagrangian density as opposed to those from earliers Secs.. When the initial Lagrangian density is a bare one at one-loop order, we introduce terms to the initial Lagrange density for eliminating the divergences for finding a finite one. We repeat this procedure to the loop order and so on. We absorb the divergences of the theory though the following renormalization constants
\begin{eqnarray}\label{uhguhfgugu}
Z_{\phi,\zeta_{a}} = 1 + \sum_{i=1}^{\infty} c_{\phi}^{i}, 
\end{eqnarray} 
\begin{eqnarray}\label{uhguhfgugud}
Z_{\zeta_{a},u} = 1 + \sum_{i=1}^{\infty} c_{u}^{i},
\end{eqnarray} 
\begin{eqnarray}\label{uhguhfgugus}
Z_{\phi^{2},\zeta_{a}} = 1 + \sum_{i=1}^{\infty} c_{\phi^{2}}^{i},
\end{eqnarray} 
where 
\begin{eqnarray}\label{huytrjii}
\mathcal{L}_{\zeta_{a}} = \frac{1}{2}Z_{\phi}\partial^{\mu}\phi\partial_{\mu}\phi + \frac{1}{2}Z_{\phi}[(\zeta_{a}\cdot\partial)\phi]^{2} + \frac{\mu^{\epsilon}u}{4!}Z_{u}\phi^{4} +  \frac{1}{2}tZ_{\phi^{2},\zeta_{a}}\phi^{2}
\end{eqnarray}
and
\begin{eqnarray}\label{huytr}
\phi = Z_{\phi,\zeta_{a}}^{-1/2}\phi_{B}, ~~~~ \lambda = \mu^{-\epsilon}\frac{Z_{\phi,\zeta_{a}}^{2}}{Z_{u,\zeta_{a}}}\lambda_{B}, ~~~~ t = \frac{Z_{\phi,\zeta_{a}}}{Z_{\phi^{2},\zeta_{a}}}t_{B}. 
\end{eqnarray} 
 
The renormalization constants, in terms of Feynman diagrams, are given by
\begin{eqnarray}\label{Zphi}
&& Z_{\phi,\zeta_{a}}(u,\epsilon^{-1}) = 1 +  \frac{1}{P^{2} + (\zeta_{a}\cdot P)^{2}} \Biggl[ \frac{1}{6} \mathcal{K} \left(\parbox{12mm}{\includegraphics[scale=1.0]{fig6.eps}} \right) S_{\parbox{10mm}{\includegraphics[scale=0.5]{fig6.eps}}}  + \frac{1}{4} \mathcal{K} \left(\parbox{12mm}{\includegraphics[scale=1.0]{fig7.eps}} \right) S_{\parbox{6mm}{\includegraphics[scale=0.5]{fig7.eps}}} + \frac{1}{3} \mathcal{K}   \left(\parbox{12mm}{\includegraphics[scale=1.0]{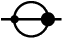}} \right) S_{\parbox{6mm}{\includegraphics[scale=0.5]{fig26.eps}}} \Biggr], \hspace{4mm}
\end{eqnarray}

\begin{eqnarray}\label{Zg}
&& Z_{u,\zeta_{a}}(u,\epsilon^{-1}) = 1 + \frac{1}{\mu^{\epsilon}u} \Biggl[ \frac{1}{2} \mathcal{K} 
\left(\parbox{10mm}{\includegraphics[scale=1.0]{fig10.eps}} + 2 \hspace{1mm} perm.
\right) S_{\parbox{8mm}{\includegraphics[scale=0.5]{fig10.eps}}}  + \frac{1}{4} \mathcal{K} 
\left(\parbox{17mm}{\includegraphics[scale=1.0]{fig11.eps}} + 2 \hspace{1mm} perm. \right) S_{\parbox{10mm}{\includegraphics[scale=0.5]{fig11.eps}}}  + \nonumber \\&& \frac{1}{2} \mathcal{K} 
\left(\parbox{12mm}{\includegraphics[scale=.8]{fig21.eps}} + 5 \hspace{1mm} perm. \right) S_{\parbox{10mm}{\includegraphics[scale=0.4]{fig21.eps}}}  + \mathcal{K}   \left(\parbox{10mm}{\includegraphics[scale=1.0]{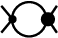}} + 2 \hspace{1mm} perm. \right) S_{\parbox{6mm}{\includegraphics[scale=0.5]{fig25.eps}}}\Biggr],
\end{eqnarray}

\begin{eqnarray}\label{Zphi2}
&& Z_{\phi^{2},\zeta_{a}}(u,\epsilon^{-1}) = 1 + \frac{1}{2} \mathcal{K} 
\left(\parbox{14mm}{\includegraphics[scale=1.0]{fig14.eps}} \right) S_{\parbox{10mm}{\includegraphics[scale=0.5]{fig14.eps}}} + \frac{1}{4} \mathcal{K} 
\left(\parbox{12mm}{\includegraphics[scale=1.0]{fig16.eps}} \right) S_{\parbox{10mm}{\includegraphics[scale=0.5]{fig16.eps}}} + \frac{1}{2} \mathcal{K} \left(\parbox{11mm}{\includegraphics[scale=.8]{fig17.eps}} \right) S_{\parbox{8mm}{\includegraphics[scale=0.4]{fig17.eps}}} + \nonumber \\ &&  \frac{1}{2} \mathcal{K}
  \left(\parbox{12mm}{\includegraphics[scale=.2]{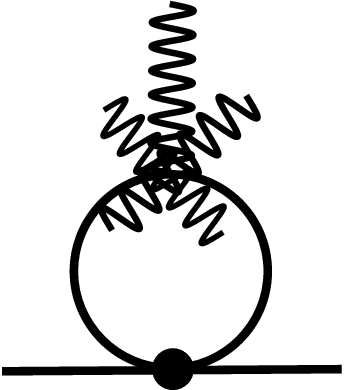}} \right) S_{\parbox{7mm}{\includegraphics[scale=.12]{fig31.eps}}} + \frac{1}{2} \mathcal{K}
  \left(\parbox{12mm}{\includegraphics[scale=.2]{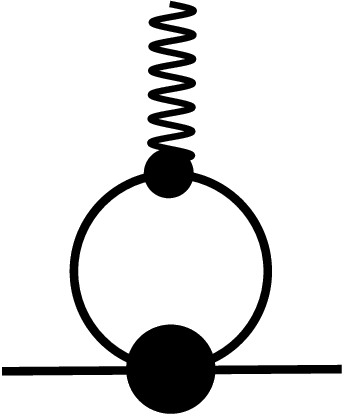}} \right) S_{\parbox{7mm}{\includegraphics[scale=.12]{fig32.eps}}},
\end{eqnarray}
where, for example, the factor $S_{\parbox{6mm}{\includegraphics[scale=0.5]{fig6.eps}}}$ is the symmetry factor for that Feynman diagram when the number $N$ of components of the field is not equal to 1. The Feynman diagrams needed in the present method are the ones \cite{Amit}
\begin{eqnarray}
\parbox{12mm}{\includegraphics[scale=1.0]{fig6.eps}} =  -\mathbf{Z}_{\zeta_{a}}^{2}\frac{u^{2}[P^{2} + (\zeta_{a}\cdot P)^{2}]}{8\epsilon} \left\{ 1 + \frac{1}{4}\epsilon -2\epsilon L_{3}[P^{2} + (\zeta_{a}\cdot P)^{2}] \right\},
\end{eqnarray}  
\begin{eqnarray}
\parbox{10mm}{\includegraphics[scale=0.9]{fig7.eps}} = \mathbf{Z}_{\zeta_{a}}^{3}\frac{u^{3}[P^{2} + (\zeta_{a}\cdot P)^{2}]}{6\epsilon^{2}} \left\{ 1 + \frac{1}{2}\epsilon -3\epsilon L_{3}[P^{2} + (\zeta_{a}\cdot P)^{2}] \right\} \mathbf{\Pi}^{3},
\end{eqnarray}   
\begin{eqnarray}
\parbox{10mm}{\includegraphics[scale=1.0]{fig10.eps}} =  \frac{\mathbf{Z}_{\zeta_{a}}^{2}\mu^{\epsilon}u^{2}}{\epsilon}  \left\{1 - \frac{1}{2}\epsilon - \frac{1}{2}\epsilon L[P^{2} + (\zeta_{a}\cdot P)^{2}] \right\} \mathbf{\Pi},
\end{eqnarray}   
\begin{eqnarray}
\parbox{12mm}{\includegraphics[scale=0.8]{fig21.eps}} =  -\frac{\mathbf{Z}_{\zeta_{a}}^{2}\mu^{\epsilon}u^{3}}{2\epsilon^{2}}  \left\{1 - \frac{1}{2}\epsilon - \epsilon L[P^{2} + (\zeta_{a}\cdot P)^{2}] \right\} \mathbf{\Pi}^{2},
\end{eqnarray}  
where
\begin{eqnarray}\label{uahuahuahlol}
L[P^{2} + (\zeta_{a}\cdot P)^{2}] =  \int_{0}^{1}dx\ln\left\{\frac{x(1-x)[P^{2} + (\zeta_{a}\cdot P)^{2}]}{\mu^{2}}\right\},
\end{eqnarray}
\begin{eqnarray}
L_{3}[P^{2} + (\zeta_{a}\cdot P)^{2}] =  \int_{0}^{1}dx(1-x)\ln\left\{\frac{x(1-x)[P^{2} + (\zeta_{a}\cdot P)^{2}]}{\mu^{2}}\right\}. 
\end{eqnarray}
We can compute the $\beta_{\zeta_{a}}$-function and anomalous dimensions from the renormalization group equation
\begin{eqnarray}\left( \mu\frac{\partial}{\partial\mu} + \beta_{\zeta_{a}}\frac{\partial}{\partial u} - \frac{1}{2}n\gamma_{\phi,\zeta_{a}} + l\gamma_{\phi^{2},\zeta_{a}} \right)\Gamma_{\zeta_{a}}^{(n,l)} = 0,
\end{eqnarray}
where 
\begin{eqnarray}\label{tjjaffaxdzs}
\beta_{\zeta_{a}}(u) = \mu\frac{\partial u}{\partial \mu},
\end{eqnarray}
\begin{eqnarray}\label{ejjaffaxdzsd}
\gamma_{\phi,\zeta_{a}}(u) = \mu\frac{\partial\ln Z_{\phi,\zeta_{a}}}{\partial \mu},
\end{eqnarray}
\begin{eqnarray}\label{djjaffaxdzse}
\gamma_{\phi^{2},\zeta_{a}}(u) = -\mu\frac{\partial\ln Z_{\phi^{2},\zeta_{a}}}{\partial \mu}.
\end{eqnarray}
Thus we obtain
\begin{eqnarray}\label{reewriretjgjk}
\beta_{\zeta_{a}}(u) = -\epsilon u + \mathbf{Z}_{\zeta_{a}}\frac{N + 8}{6} u^{2} - \mathbf{Z}_{\zeta_{a}}^{2}\frac{3N + 14}{12},
\end{eqnarray} 
\begin{eqnarray}\label{jkjkpfgjrftj}
\gamma_{\phi,\zeta_{a}}(u) = \mathbf{Z}_{\zeta_{a}}^{2}\frac{N + 2}{72}u^{2} - \mathbf{Z}_{\zeta_{a}}^{3}\frac{(N + 2)(N + 8)}{1728}u^{3},
\end{eqnarray} 
\begin{eqnarray}\label{gfydsguyfsdgufa}
\gamma_{\phi^{2},\zeta_{a}}(u) = \mathbf{Z}_{\zeta_{a}}\frac{N + 2}{6} u - \mathbf{Z}_{\zeta_{a}}^{2}\frac{5(N + 2)}{72}u^{2}.
\end{eqnarray} 
By computing the nontrivial fixed point and obtaining 
\begin{eqnarray}
u_{\zeta_{a}}^{*} = \frac{6\epsilon/\mathbf{Z}_{\zeta_{a}}}{(N + 8)}\left\{ 1 + \epsilon\left[ \frac{3(3N + 14)}{(N + 8)^{2}} \right]\right\}
\end{eqnarray}
and applying the relations $\eta_{\zeta_{a}}\equiv\gamma_{\phi,\zeta_{a}}(u^{*})$ and $\nu_{\zeta_{a}}^{-1}\equiv 2 - \gamma_{\phi^{2},\zeta_{a}}(u_{\zeta_{a}}^{*})$, we obtain that the critical exponents values obtained are the same as those from the Euclidean theory \cite{Wilson197475,PhysRevLett.28.240,PhysRevLett.28.548}. Now we have to generalize the results obtained NLO to any loop order.

\section{Critical exponents to any loop order}\label{Critical exponents to any loop order}

\par We can compute the critical exponents values for any loop level. For that we have to employ the following theorem:  
\begin{theorem} 
For some momentum space general $1$PI Feynman diagram of arbitrary loop order in a theory whose Lagrangian density is that of Eq. (\ref{huytrji}), we can compute its value in dimensional regularization in $d = 4 - \epsilon$ and write the corresponding result as a general functional $\mathbf{Z}_{\zeta_{a}}^{L}\mathcal{F}(u,P^{2} + (\zeta_{a}\cdot P)^{2},\epsilon,\mu)$ if its Euclidean space-time general functional counterpart is given by $\mathcal{F}(u,P^{2},\epsilon,\mu)$. The number of loops of the referred diagram is $L$.
\end{theorem}

\begin{proof} 
In some arbitrary Feynman diagram of loop order $L$, we have to compute a multidimensional integral for $L$ momentum variables $q_{1}$, $q_{2}$,...,$q_{L}$. The volume element associated to each momentum variable is given by $d^{d}q_{i}$ ($i = 1, 2,...,L$). In Sec. \ref{Normalization conditions method}, we have shown that the variable change $q^{\prime} \rightarrow \sqrt{\mathbb{I} + \zeta_{a}\zeta_{a}^{t}}~q$ transforms a given volume element into $d^{d}q^{\prime} = \sqrt{det(\mathbb{I} + \zeta_{a}\zeta_{a}^{t})}d^{d}q$. Thus $d^{d}q = d^{d}q^{\prime}/\sqrt{det(\mathbb{I} + \zeta_{a}\zeta_{a}^{t})} \equiv \mathbf{Z}_{\zeta_{a}}d^{d}q^{\prime}$. So a power of $\mathbf{Z}_{\zeta_{a}}$ emerges from each momentum integration. As there are $L$ integrals, we obtain that each diagram can be written as $\mathbf{Z}_{\zeta_{a}}^{L}\mathcal{F}(u,P^{2} + (\zeta_{a}\cdot P)^{2},\epsilon,\mu)$, where $\mathcal{F}(u,P^{2},\epsilon,\mu)$ is the corresponding computed Feynman diagram result for the Euclidean space-time. 
\end{proof} 

\par Now without loss of generality, we apply the method of last Sec., as the critical exponents obtained through any renormalization scheme are the same. As the theory in Euclidean space-time is renormalizable for all loop orders \cite{Amit}, we can employ the theorem of the present Sec. and write any Feynman diagram of the present theory as $\mathbf{Z}_{\zeta_{a}}^{L}\mathcal{F}(u,P^{2} + (\zeta_{a}\cdot P)^{2},\epsilon,\mu)$, where $\mathcal{F}(u,P^{2},\epsilon,\mu)$ is the corresponding computed Feynman diagram result for the Euclidean space-time.  Now by applying the BPHZ theorem \cite{BogoliubovParasyuk,Hepp,Zimmermann}, which asserts that all momentum-dependent integrals are canceled through the renormalization program, order by order in perturbation theory. Thus the only dependence of the $\beta_{\zeta_{a}}$-function and anomalous dimensions on $\zeta_{a}$ comes from the exact $\mathbf{Z}_{\zeta_{a}}$ factor rised to a some power, namely to $L$. Then the exact effect of the Randers-Finsler space-time on the $\beta_{\zeta_{a}}$-function and anomalous dimensions for all loop orders can be written as
\begin{eqnarray}\label{uhgufhduhufdhu}
\beta_{\zeta_{a}}(u) =  -\epsilon u + \sum_{n=2}^{\infty}\mathbf{Z}_{\zeta_{a}}^{n-1}\beta_{n}^{(0)}u^{n}, 
\end{eqnarray}
\begin{eqnarray}
\gamma_{\phi,\zeta_{a}}(u) = \sum_{n=2}^{\infty}\mathbf{Z}_{\zeta_{a}}^{n}\gamma_{n}^{(0)}u^{n},
\end{eqnarray}
\begin{eqnarray}
\gamma_{\phi^{2},\zeta_{a}}(u) = \sum_{n=1}^{\infty}\mathbf{Z}_{\zeta_{a}}^{n}\gamma_{\phi^{2}, n}^{(0)}u^{n},
\end{eqnarray}
where $\beta_{n}^{(0)}$, $\gamma_{n}^{(0)}$ and $\gamma_{\phi^{2}, n}^{(0)}$ are the corresponding nth-loop corrections to the referred functions in Euclidean space-time \cite{Wilson197475,PhysRevLett.28.240,PhysRevLett.28.548}. Now by computing the nontrivial fixed point of the $\beta_{\zeta_{a}}$-function of Eq. (\ref{uhgufhduhufdhu}) valid for all-loop level, we obtain $u_{\zeta_{a}}^{*} = u^{*}/\mathbf{Z}_{\zeta_{a}}$, where $u^{*}$ is the all-loop order nontrivial fixed point for the Euclidean space-time theory. By employing the relations $\eta_{\zeta_{a}}\equiv\gamma_{\phi}(u^{*})$ and $\nu_{\zeta_{a}}^{-1}\equiv 2 - \gamma_{\phi^{2},\zeta_{a}}(u_{\zeta_{a}}^{*})$, we obtain critical exponents as the same as their Euclidean space-time counterparts.

\section{Conclusions}\label{Conclusions}

\par We have studied the scaling behavior of Randers-Finsler massless scalar field theories in the infrared regime. For that, we have computed analytically the radiative corrections to the corresponding anomalous dimensions, related to the critical exponents of the theory, first up to NLO and later, for some induction process, for all-loop levels. We have considered the effect of the Randers-Finsler space-time properties on the critical exponents by considering the parameter characterizing those space-times in its exact form, namely $\zeta_{a}$. We have employed field-theoretic renormalization group and $\epsilon$-expansion techniques at dimensions $d = 4 - \epsilon$ through three distinct and independent methods. The physical interpretation of the obtained results was that the properties we are considering are that of the space-time where the field is embedded and they do not modify how the field field interacts with itself (which would lead to a modification of the critical exponents values), although the $\beta_{\zeta_{a}}$-function and anomalous dimensions depend on $\zeta_{a}$. This result is in accordance with the universality hypothesis, which asserts that the critical exponents values depend only on the dimension $d$, number $N$ of some $N$-component order parameter and the interactions among the constituents of the system. This work can motivate further studies on effects at Randers-Finsler space-times.

\section{Acknowledgements}

\par PRSC would like to thank the Brazilian funding agencies CAPES and CNPq (Grant: Produtividade 306130/2022-0) for financial support.

\bibliography{apstemplate}

\providecommand{\noopsort}[1]{}\providecommand{\singleletter}[1]{#1}%
\begin{thebibliography}{18}%
\makeatletter
\providecommand \@ifxundefined [1]{%
 \@ifx{#1\undefined}
}%
\providecommand \@ifnum [1]{%
 \ifnum #1\expandafter \@firstoftwo
 \else \expandafter \@secondoftwo
 \fi
}%
\providecommand \@ifx [1]{%
 \ifx #1\expandafter \@firstoftwo
 \else \expandafter \@secondoftwo
 \fi
}%
\providecommand \natexlab [1]{#1}%
\providecommand \enquote  [1]{``#1''}%
\providecommand \bibnamefont  [1]{#1}%
\providecommand \bibfnamefont [1]{#1}%
\providecommand \citenamefont [1]{#1}%
\providecommand \href@noop [0]{\@secondoftwo}%
\providecommand \href [0]{\begingroup \@sanitize@url \@href}%
\providecommand \@href[1]{\@@startlink{#1}\@@href}%
\providecommand \@@href[1]{\endgroup#1\@@endlink}%
\providecommand \@sanitize@url [0]{\catcode `\\12\catcode `\$12\catcode
  `\&12\catcode `\#12\catcode `\^12\catcode `\_12\catcode `\%12\relax}%
\providecommand \@@startlink[1]{}%
\providecommand \@@endlink[0]{}%
\providecommand \url  [0]{\begingroup\@sanitize@url \@url }%
\providecommand \@url [1]{\endgroup\@href {#1}{\urlprefix }}%
\providecommand \urlprefix  [0]{URL }%
\providecommand \Eprint [0]{\href }%
\providecommand \doibase [0]{http://dx.doi.org/}%
\providecommand \selectlanguage [0]{\@gobble}%
\providecommand \bibinfo  [0]{\@secondoftwo}%
\providecommand \bibfield  [0]{\@secondoftwo}%
\providecommand \translation [1]{[#1]}%
\providecommand \BibitemOpen [0]{}%
\providecommand \bibitemStop [0]{}%
\providecommand \bibitemNoStop [0]{.\EOS\space}%
\providecommand \EOS [0]{\spacefactor3000\relax}%
\providecommand \BibitemShut  [1]{\csname bibitem#1\endcsname}%
\let\auto@bib@innerbib\@empty
\bibitem [{\citenamefont {Bejancu}\ and\ \citenamefont
  {Farran}(2010)}]{Bejancu}%
  \BibitemOpen
  \bibfield  {author} {\bibinfo {author} {\bibfnamefont {A.}~\bibnamefont
  {Bejancu}}\ and\ \bibinfo {author} {\bibfnamefont {H.~R.}\ \bibnamefont
  {Farran}},\ }\href@noop {} {\emph {\bibinfo {title} {Geometry of
  Pseudo-Finsler Submanifolds}}}\ (\bibinfo  {publisher} {Springer},\ \bibinfo
  {year} {2010})\BibitemShut {NoStop}%
\bibitem [{\citenamefont {Hohmann}\ \emph {et~al.}(2022)\citenamefont
  {Hohmann}, \citenamefont {Pfeifer},\ and\ \citenamefont
  {Voicu}}]{10.1063.5.0065944}%
  \BibitemOpen
  \bibfield  {author} {\bibinfo {author} {\bibfnamefont {M.}~\bibnamefont
  {Hohmann}}, \bibinfo {author} {\bibfnamefont {C.}~\bibnamefont {Pfeifer}}, \
  and\ \bibinfo {author} {\bibfnamefont {N.}~\bibnamefont {Voicu}},\ }\href
  {\doibase 10.1063/5.0065944} {\bibfield  {journal} {\bibinfo  {journal} {J.
  Math. Phys.}\ }\textbf {\bibinfo {volume} {63}},\ \bibinfo {pages} {032503}
  (\bibinfo {year} {2022})}\BibitemShut {NoStop}%
\bibitem [{\citenamefont {Silva}(2021)}]{Silva2021}%
  \BibitemOpen
  \bibfield  {author} {\bibinfo {author} {\bibfnamefont {J.~E.~G.}\
  \bibnamefont {Silva}},\ }\href {\doibase 10.1209/0295-5075/133/21002}
  {\bibfield  {journal} {\bibinfo  {journal} {EPL}\ }\textbf {\bibinfo {volume}
  {133}},\ \bibinfo {pages} {21002} (\bibinfo {year} {2021})}\BibitemShut
  {NoStop}%
\bibitem [{\citenamefont {Vacaru}(2010)}]{VACARU2010224}%
  \BibitemOpen
  \bibfield  {author} {\bibinfo {author} {\bibfnamefont {S.~I.}\ \bibnamefont
  {Vacaru}},\ }\href {\doibase https://doi.org/10.1016/j.physletb.2010.05.036}
  {\bibfield  {journal} {\bibinfo  {journal} {Phys. Lett. B}\ }\textbf
  {\bibinfo {volume} {690}},\ \bibinfo {pages} {224} (\bibinfo {year}
  {2010})}\BibitemShut {NoStop}%
\bibitem [{\citenamefont {Praveen}\ and\ \citenamefont
  {Narasimhamurthy}(2025)}]{PRAVEEN2025100030}%
  \BibitemOpen
  \bibfield  {author} {\bibinfo {author} {\bibfnamefont {J.}~\bibnamefont
  {Praveen}}\ and\ \bibinfo {author} {\bibfnamefont {S.}~\bibnamefont
  {Narasimhamurthy}},\ }\href {\doibase
  https://doi.org/10.1016/j.jspc.2025.100030} {\bibfield  {journal} {\bibinfo
  {journal} {JSPC}\ }\textbf {\bibinfo {volume} {3}},\ \bibinfo {pages}
  {100030} (\bibinfo {year} {2025})}\BibitemShut {NoStop}%
\bibitem [{\citenamefont {Kapsabelis}\ \emph {et~al.}(2024)\citenamefont
  {Kapsabelis}, \citenamefont {Saridakis},\ and\ \citenamefont
  {Stavrinos}}]{KapsabelisE}%
  \BibitemOpen
  \bibfield  {author} {\bibinfo {author} {\bibfnamefont {E.}~\bibnamefont
  {Kapsabelis}}, \bibinfo {author} {\bibfnamefont {E.}~\bibnamefont
  {Saridakis}}, \ and\ \bibinfo {author} {\bibfnamefont {P.}~\bibnamefont
  {Stavrinos}},\ }\href {\doibase 10.1103/PhysRevD.88.123510} {\bibfield
  {journal} {\bibinfo  {journal} {Eur. Phys. J. C.}\ }\textbf {\bibinfo
  {volume} {84}},\ \bibinfo {pages} {538} (\bibinfo {year} {2024})}\BibitemShut
  {NoStop}%
\bibitem [{\citenamefont {Basilakos}\ \emph {et~al.}(2013)\citenamefont
  {Basilakos}, \citenamefont {Kouretsis}, \citenamefont {Saridakis},\ and\
  \citenamefont {Stavrinos}}]{PhysRevD.88.123510}%
  \BibitemOpen
  \bibfield  {author} {\bibinfo {author} {\bibfnamefont {S.}~\bibnamefont
  {Basilakos}}, \bibinfo {author} {\bibfnamefont {A.~P.}\ \bibnamefont
  {Kouretsis}}, \bibinfo {author} {\bibfnamefont {E.~N.}\ \bibnamefont
  {Saridakis}}, \ and\ \bibinfo {author} {\bibfnamefont {P.~C.}\ \bibnamefont
  {Stavrinos}},\ }\href {\doibase 10.1103/PhysRevD.88.123510} {\bibfield
  {journal} {\bibinfo  {journal} {Phys. Rev. D}\ }\textbf {\bibinfo {volume}
  {88}},\ \bibinfo {pages} {123510} (\bibinfo {year} {2013})}\BibitemShut
  {NoStop}%
\bibitem [{\citenamefont {Basilakos}\ and\ \citenamefont
  {Stavrinos}(2013)}]{PhysRevD.87.043506}%
  \BibitemOpen
  \bibfield  {author} {\bibinfo {author} {\bibfnamefont {S.}~\bibnamefont
  {Basilakos}}\ and\ \bibinfo {author} {\bibfnamefont {P.}~\bibnamefont
  {Stavrinos}},\ }\href {\doibase 10.1103/PhysRevD.87.043506} {\bibfield
  {journal} {\bibinfo  {journal} {Phys. Rev. D}\ }\textbf {\bibinfo {volume}
  {87}},\ \bibinfo {pages} {043506} (\bibinfo {year} {2013})}\BibitemShut
  {NoStop}%
\bibitem [{\citenamefont {Chang}\ and\ \citenamefont
  {Li}(2009)}]{CHANG2009173}%
  \BibitemOpen
  \bibfield  {author} {\bibinfo {author} {\bibfnamefont {Z.}~\bibnamefont
  {Chang}}\ and\ \bibinfo {author} {\bibfnamefont {X.}~\bibnamefont {Li}},\
  }\href {\doibase https://doi.org/10.1016/j.physletb.2009.05.001} {\bibfield
  {journal} {\bibinfo  {journal} {Phys. Lett. B}\ }\textbf {\bibinfo {volume}
  {676}},\ \bibinfo {pages} {173} (\bibinfo {year} {2009})}\BibitemShut
  {NoStop}%
\bibitem [{\citenamefont {Zinn-Justin}(2002)}]{ZinnJustin}%
  \BibitemOpen
  \bibfield  {author} {\bibinfo {author} {\bibfnamefont {J.}~\bibnamefont
  {Zinn-Justin}},\ }\href@noop {} {\emph {\bibinfo {title} {Quantum Field
  Theory and Critical Phenomena}}}\ (\bibinfo  {publisher} {International
  Series of Monographs on Physics, Oxford University Press},\ \bibinfo {year}
  {2002})\BibitemShut {NoStop}%
\bibitem [{\citenamefont {Amit}\ and\ \citenamefont
  {Mart\'in-Mayor}(2005)}]{Amit}%
  \BibitemOpen
  \bibfield  {author} {\bibinfo {author} {\bibfnamefont {D.~J.}\ \bibnamefont
  {Amit}}\ and\ \bibinfo {author} {\bibfnamefont {V.}~\bibnamefont
  {Mart\'in-Mayor}},\ }\href@noop {} {\emph {\bibinfo {title} {Field Theory,
  The Renormalization Group and Critical Phenomena}}}\ (\bibinfo  {publisher}
  {World Scientific Pub Co Inc},\ \bibinfo {year} {2005})\BibitemShut {NoStop}%
\bibitem [{\citenamefont {Pelissetto}\ and\ \citenamefont
  {Vicari}(2002)}]{Pelissetto2002549}%
  \BibitemOpen
  \bibfield  {author} {\bibinfo {author} {\bibfnamefont {A.}~\bibnamefont
  {Pelissetto}}\ and\ \bibinfo {author} {\bibfnamefont {E.}~\bibnamefont
  {Vicari}},\ }\href@noop {} {\bibfield  {journal} {\bibinfo  {journal} {Phys.
  Rep.}\ }\textbf {\bibinfo {volume} {368}},\ \bibinfo {pages} {549} (\bibinfo
  {year} {2002})}\BibitemShut {NoStop}%
\bibitem [{\citenamefont {Wilson}\ and\ \citenamefont
  {Kogut}(1974)}]{Wilson197475}%
  \BibitemOpen
  \bibfield  {author} {\bibinfo {author} {\bibfnamefont {K.~G.}\ \bibnamefont
  {Wilson}}\ and\ \bibinfo {author} {\bibfnamefont {J.}~\bibnamefont {Kogut}},\
  }\href@noop {} {\bibfield  {journal} {\bibinfo  {journal} {Phys. Rep.}\
  }\textbf {\bibinfo {volume} {12}},\ \bibinfo {pages} {75} (\bibinfo {year}
  {1974})}\BibitemShut {NoStop}%
\bibitem [{\citenamefont {Wilson}\ and\ \citenamefont
  {Fisher}(1972)}]{PhysRevLett.28.240}%
  \BibitemOpen
  \bibfield  {author} {\bibinfo {author} {\bibfnamefont {K.~G.}\ \bibnamefont
  {Wilson}}\ and\ \bibinfo {author} {\bibfnamefont {M.~E.}\ \bibnamefont
  {Fisher}},\ }\href@noop {} {\bibfield  {journal} {\bibinfo  {journal} {Phys.
  Rev. Lett.}\ }\textbf {\bibinfo {volume} {28}},\ \bibinfo {pages} {240}
  (\bibinfo {year} {1972})}\BibitemShut {NoStop}%
\bibitem [{\citenamefont {Wilson}(1972)}]{PhysRevLett.28.548}%
  \BibitemOpen
  \bibfield  {author} {\bibinfo {author} {\bibfnamefont {K.~G.}\ \bibnamefont
  {Wilson}},\ }\href@noop {} {\bibfield  {journal} {\bibinfo  {journal} {Phys.
  Rev. Lett.}\ }\textbf {\bibinfo {volume} {28}},\ \bibinfo {pages} {548}
  (\bibinfo {year} {1972})}\BibitemShut {NoStop}%
\bibitem [{\citenamefont {Bogoliubov}\ and\ \citenamefont
  {Parasyuk}(1957)}]{BogoliubovParasyuk}%
  \BibitemOpen
  \bibfield  {author} {\bibinfo {author} {\bibfnamefont {N.~N.}\ \bibnamefont
  {Bogoliubov}}\ and\ \bibinfo {author} {\bibfnamefont {O.~S.}\ \bibnamefont
  {Parasyuk}},\ }\href@noop {} {\bibfield  {journal} {\bibinfo  {journal} {Acta
  Math.}\ }\textbf {\bibinfo {volume} {97}},\ \bibinfo {pages} {227} (\bibinfo
  {year} {1957})}\BibitemShut {NoStop}%
\bibitem [{\citenamefont {Hepp}(1966)}]{Hepp}%
  \BibitemOpen
  \bibfield  {author} {\bibinfo {author} {\bibfnamefont {K.}~\bibnamefont
  {Hepp}},\ }\href@noop {} {\bibfield  {journal} {\bibinfo  {journal} {Commun.
  Math. Phys.}\ }\textbf {\bibinfo {volume} {2}},\ \bibinfo {pages} {301}
  (\bibinfo {year} {1966})}\BibitemShut {NoStop}%
\bibitem [{\citenamefont {Zimmermann}(1969)}]{Zimmermann}%
  \BibitemOpen
  \bibfield  {author} {\bibinfo {author} {\bibfnamefont {W.}~\bibnamefont
  {Zimmermann}},\ }\href@noop {} {\bibfield  {journal} {\bibinfo  {journal}
  {Commun. Math. Phys.}\ }\textbf {\bibinfo {volume} {15}},\ \bibinfo {pages}
  {208} (\bibinfo {year} {1969})}\BibitemShut {NoStop}%
\end{thebibliography}%

\end{document}